\documentclass{article}
\usepackage{fullpage}

\usepackage{booktabs} % For formal tables
\usepackage[ruled]{algorithm2e} % For algorithms

\SetAlFnt{\small}
\SetAlCapFnt{\small}
\SetAlCapNameFnt{\small}
\SetAlCapHSkip{0pt}
\IncMargin{-\parindent}

\title{Best vs. All:\\ Equity and Accuracy of Standardized Test Score Reporting}

\usepackage{amsmath,amsfonts,amsthm}
\usepackage{mathtools}
\usepackage{float}
\usepackage{bbm}
\usepackage{bm}
\usepackage{epsfig}
\usepackage{graphicx,hyperref}
\usepackage{xcolor}
\usepackage{amssymb}

\newtheorem{theorem}{Theorem}

\newtheorem{defin}{ {Definition}}[section]

\newtheorem{lemma}[defin]{\bf {Lemma}}
\newtheorem{corollary}[defin]{\bf {Corollary}}

\newcommand{\remark}{{\bf Remark: }}

\newcommand{\E}{\mathbb{E}}
\newcommand{\st}{{\hbox{ s.t. }}}

\newcommand{\Xomit}[1]{}

\newcommand{\thanhcomment}{\color{purple}}
\newcommand{\ar}[1]{\textcolor{red}{[Aaron: #1]}}
\newcommand{\mn}[1]{\textcolor{blue}{[Mingzi: #1]}}
\newcommand{\sk}[1]{\textcolor{blue}{[Sampath: #1]}}

\author{Sampath Kannan\thanks{University of Pennsylvania, Department of Computer and Information Sciences} \and Mingzi Niu\thanks{Rice University, Department of Economics} \and Aaron Roth\thanks{University of Pennsylvania, Department of Computer and Information Sciences} \and Rakesh Vohra\thanks{University of Pennsylvania, Department of Economics and Electrical and Systems Engineering}}

\begin{document}
%\begin{titlepage}
\maketitle
%\end{titlepage}
\begin{abstract}
We study a game theoretic model of standardized testing for college admissions. Students are of two types; High and Low. There is a college that would like to admit the High type students. Students take a potentially costly standardized exam which provides a noisy signal of their type. 

The students come from two populations, which are identical in talent (i.e. the type distribution is the same), but differ in their access to resources: the higher resourced population can at their option take the exam multiple times, whereas the lower resourced population can only take the exam once. We study two models of score reporting, which capture existing policies used by colleges. The first policy (sometimes known as ``super-scoring'') allows students to report the \emph{max} of the scores they achieve. The other policy requires that all scores be reported. 

We find in our model that requiring that all scores be reported results in superior outcomes in equilibrium, both from the perspective of the college (the admissions rule is more accurate), \emph{and} from the perspective of equity across populations: a student's probability of admission is independent of their population, conditional on their type. In particular, the false positive rates and false negative rates are identical in this setting, across the highly and poorly resourced student populations. This is the case despite the fact that the more highly resourced students can---at their option---either report a more accurate signal of their type, or pool with the lower resourced population under this policy. 
\end{abstract}

% \doublespacing
\section{Introduction}
Worldwide, standardized tests are an essential input to University admission decisions. Their use is considered to benefit both institutions and individuals. Institutions want to avoid the false-positive error of selection so as not to admit those who will not thrive. False-negative errors (the rejection of would-be successes) deprive the institution, and hence society, of potential college graduates. They are also a harm to individual students, as they  prevent a potentially successful student from attending University. 

Given the importance of standardized tests, there is an extensive literature on the fairness of such tests and how to compensate for possible unfairness. Thorndike, for example, pointed out in 1971 that the distribution of false negatives between groups and not just their absolute numbers matters. His proposed remedy involved setting two selection cutoffs with a lower cutoff for the members of the minority group \cite{thorndike}.

In this paper we consider a source of unfairness that arises from the fact that some---but not all---applicants have the resources to take the test a number of times and have discretion over what scores they submit. Until 2008, the custom was for US Universities to ask for scores from all attempts on either the SAT or the ACT and these were provided by the relevant testing companies. In 2008, the SAT introduced the superscoring option.\footnote{\url{https://www.nytimes.com/2008/12/31/education/31sat.html}} Applicant's can pick and choose which of their SAT scores to submit. In 2020, the ACT went a step further and allowed students to retake individual sections of a test (a test has multiple sections, eg, mathematics, verbal, writing, etc). Furthermore, it will include the superscore---i.e. the component-wise max---in all ACT score reports.\footnote{\url{https://www.act.org/content/act/en/new-act-options/superscoring.html}} 

In turn, many Universities have switched from a policy of requiring all scores to be reported.\footnote{In part this is driven by competitive pressure. See \url{ https://youtu.be/yFm5fAt0zIo} for discussion between Louisiana State University's VP for enrollment management and the University's Board of Supervisors held on December 05, 2019.}   At present Universities employ one of the following  policies:
\begin{enumerate}
\item  Require all scores:  applicants must submit all test scores.
%\item Recommend all scores: applicants are urged but not required to submit all test scores. In these cases the University may also commit to how these scores will be used. For example, if one takes the SAT more than once, the admission office will take the highest section score across tests (i.e. the admissions office will themselves implement super-scoring).
%\sk{The effective difference between this option and option 3 is not coming across... should we only have two options?}
\item  Score choice: the applicant is free to choose which scores to submit. Some universities encourage students to submit all scores but commit to \emph{using} only the highest reported scores, which is equivalent.
\end{enumerate}

Georgetown University is one of the shrinking number of Universities that requires all scores. In 2019, their Dean of admissions said\footnote{\url{https://www.insidehighered.com/admissions/views/2019/12/02/should-colleges-require-students-submit-every-sat-and-act-they-take}}:
\begin{quote}
 ``If you take the SAT five times and score 600-650 on verbal on four of them but 750 on one, that is useful information compared with allowing the student to cherry-pick their best score.'' 
\end{quote}
%The Georgetown Dean frames this as an ethical issue. Just as a student does not have the right to selectively disclose their grade they shouldn't be able to selectively disclose their standardized test scores.

Superscoring  raises a fairness concern in that not all groups retest at the same rates. The ACT, for example, reports that the retest rates for Hispanic students is 34\%, compared to 49\% for both White and Asian students.  Students whose parents did not attend college had a retest rate of 36\% compared to 62\% for students whose highest parental education level exceeded a bachelor’s degree.\footnote{\url{https://www.act.org/content/dam/act/unsecured/documents/R1774-superscoring-subgroup-2019-07.pdf?_ga=2.200218796.349081470.1575775334-387505566.1575775334}} This is in spite of subsidies offered to low income students by the testing companies.\footnote{For a more detailed analysis see \url{https://www.act.org/content/dam/act/unsecured/documents/5195-Multiple-Testers.pdf}.}

%Test scores are one component of an applicant's portfolio. In some countries applicants have great latitude in the information they can submit to Universities, including, for example, extra-curricular activities, internships, and performance in academic contests. Like standardized tests these can also be thought of as delivering noisy signals of the applicant's intrinsic qualities (i.e. type). Hence, an applicant is in a position to generate noisy signals of their type from a variety of sources and choose which ones to disclose.  
In this paper, we asks what happens if applicants differ in their ability to access signals like standardized tests. In particular how are false negative and false positive rates affected in different populations, and how can policy improve these effects?

\subsection{Our Model and Results}

To answer this question we propose a simple model in which students are of two types, High ($H$) or Low ($L$). The probability that a student is of type $H$ is $p < 0.5$, capturing the idea that high types are scarce (if $p > 0.5$, the College would accept all students even if there was no testing). Students know their type, but cannot credibly convey it except through a test. The test generates a score $s \in \{A, B\}$, where $A$ stands for above the bar and $B$ for below the bar. The probability that an $H$ type student generates the score $A$ is $\alpha > 0.5$. Similarly, an $L$ type student generates a score of $B$ with probability $\alpha$.

Student's belong to one of two categories. Category 1 students are only able to take the test once. The proportion of category 1 students is $\phi$. Category 2 students, however can (in our baseline model) take the test up to twice.\footnote{This is the recommendation of a number of College coaching services which was determined by entering the following query into a search engine: ``How many times should I take the SAT?''.} In Section \ref{moretwo} we extend our results to the case in which category $2$ students can take the test $k$ times, for an arbitrary $k > 1$. This reflects the idea that students differ in their ability to access multiple signals. We might imagine, for example, that ``Category 2'' students come from a wealthier demographic. We emphasize that students in Category 2 can make their testing decisions \emph{adaptively} after observing the outcomes of their previous tests. So, for example, a student could decide to take the test a second time if their first score was a $B$, but to decline to retest if their first score was an $A$.

In our model, neither $p$ nor $\alpha$ depend on what category a student belongs to. Hence, category 1 and 2 students are ex-ante identical and the test itself is fair\footnote{We want to emphasize that while several sources of potential unfairness of standardized tests have been identified in the literature, we focus exclusively on unfairness arising from the choice of reporting policy.}---the only distinction between populations is the resources that they have available to be able to take the test twice (or not).

Scores on the test are submitted to a college, which will decide whether to accept or reject a student based on the scores submitted. The college receives a payoff of 1 for admitting a type $H$ student and -1 for admitting a type $L$. Changing these payoffs only affects the magnitudes of various cutoffs but not the qualitative conclusions. 
We consider two policies. The first requires that students report a single score. Hence, category 2 students can take the test up to twice and report their best score. We call this policy ``Report Max'' and it is akin to super scoring. The second policy requires students to report all scores. We call this ``Report All''.\footnote{It assumes students can't conceal scores.} The College can observe the reported test scores, but \emph{cannot} observe the Category that the student comes from. Thus, under Report Max, the College has no information at all about whether particular students come from Category 1 or Category 2. Under ``Report All'', if the College is sent a single test score, it has no information about which Category a student belongs to --- but if it receives multiple test scores, it knows the student must be from Category 2.

The two policies affect the incentives of Category 2 students only. Under ``Report Max'', it is intuitive that Category 2 students will take the test  twice and report their best score. Under ``Report All'', the trade-offs are more complicated. Category 2 type $H$ students, for example, who initially score an $A$ might choose to take the test a second time to separate themselves from Category 2 type $L$ students. Category 2 type $L$ students who score an $A$ on their first attempt face a choice. Stop and mimic Category 1 type $H$ students or take the test a second time to mimic Category 2 type $H$ students. These actions will affect the beliefs that the College forms given the reported scores, and therefore its admissions policy. 

In certain parameter regimes, both policies admit trivial equilibria of the form everyone is rejected (when $p$ is sufficiently small) or everyone is accepted  (when $p$ is sufficiently large). So, we restrict attention to non-trivial pure strategy equilibria only. 
These occur for values of $p \in [{\hat p}, 0.5)$ for some cut-off ${\hat p}$. 

For $p \in [{\hat p}, 0.5)$ there is a unique non-trivial pure strategy equilibrium under ``Report Max''. Category 2 students take the test twice and report their best score. The College accepts any student who reports an $A$ score and rejects all others. The effect is that Category 2 students have a different (better) signal distribution at their disposal, and hence the admissions statistics are biased in favor of Category 2 students, with lower false negative rates, and higher false positive rates. In other words, both low and high type students from Category 2 are more likely to be offered admissions compared to their Category 1 counterparts. 

Under ``Report All", students must report all test scores, and the equilibrium outcomes are not obvious. Category 2 students appear to continue to have an advantage; they no longer have access to a signal distribution that is clearly ``better'' --- but  they have access to a  more \emph{informative} signal distribution, which results from taking the test twice and reporting both outcomes. One might expect high types to take advantage of this. However they also now have the strategic option of attempting to pool with Category 1 students; one might expect low types to take advantage of this. What we show however, is that for  $p \in [{\hat p}, 0.5)$ there is a unique equilibrium outcome where Category 2 students have no advantage because the college accepts any student whose first (or only) score is $A$, and rejects those whose first (or only) score is $B$. The Category 2 students might take the test either once or twice --- but only the first score matters. This makes the equilibrium outcome entirely symmetric for both Category 1 and Category 2 students (since only a single score is relevant for admissions, for both types). Hence, ``Report All'' is superior in our model to ``Report Max'' from the perspective of equity: the false positive and false negative rates are equal across both categories of students. In other words, a student's probability of admissions, conditional on her type, is independent of her category.  The positive predictive value of ``Report All'' exceeds that of ``Report Max'', meaning that the admitted class will have a higher proportion of High types. And the expected payoff to the College is also higher under ``Report All'' than `Report Max'' --- and hence ``Report All'' is better not just from the perspective of student equity, but from the College's perspective as well. These results generalize to $k \geq 2$ tests as well, representing an unusual but fortunate setting in which the goals of equity and accuracy are aligned. We remark that an alternative solution would be to require that all students take the exam only once---but this would be (at best) challenging to implement, because standardized tests are administered by independent entities with their own interests, and different colleges have different admissions policies. What we show is that the traditional ``Report all'' policy has an equilibrium that has exactly the same effect as enforcing that students test only once. This equilibrium is unique when $k=2$, and even when $k > 2$, \emph{all} equilibria are strictly preferable under Report All, both from the perspective of the College, and from the perspective of equity (difference between false negative rates across populations) compared to the Report Max equilibrium.

%but \textcolor{blue}{the equilibrium outcome} is that the college
%'s equilibrium strategy ignores the second score and 
%accepts any student whose first (or only) score is $A$ \textcolor{blue}{
%and rejects those whose first (or only) score is $B$.
%} 
%In a pure strategy equilibrium category 2 students take the test \textcolor{blue}{either once or} twice. 

%\begin{enumerate}
%    \item The false negative rate under ``Report Max'' is strictly lower than under ``Report All''. This difference occurs entirely on the category 2 students. 
%    \item The false positive rate under ``Report Max'' is strictly higher than under ``Report All''. This difference occurs entirely on the category 2 students. 

%\item ``Report All'' has equal false positive and false negative rates across categories. ``Report Max'' does not.
%\item The positive predictive value under ``Report Max'' is strictly lower than under ``Report All''.
%\item  \textcolor{blue}{These differences between the two schemes decrease with the precision of the signal. When the signal is perfect (i.e. $\alpha = 1$), the difference vanishes.}
%\end{enumerate}

%Thus, in the relevant range ``Report All'' is  
%better than ``Report Max'' both in terms of fairness across categories and  
%the expected quality of students admitted. The only  
%disadvantage of ``Report All'' compared to ``Report Max'' is the higher false negative rate for for  category 2 students.

An interesting feature of the equilibrium outcome in ``Report All'' is that while the College asks for all tests it conditions it's admissions decisions on the first test result only. This is a {\em best response} for the College given the behavior of students and not a matter of commitment. Thus, Category 1 and 2 students are treated fairly. In conditioning on the first test score only it appears that the College is throwing away information that it could have used to improve its ability to distinguish between types. It is just that {\em in equilibrium} the second test score is not informative.

There is a half century debate about the best method for treating multiple scores from the ACT, SAT, and LSAT, see \cite{boldt} for a survey. The  focus has been on the strength of the relationship between various aggregates of (average, maximum) test scores and future GPA. That focus has, as far as we are aware, continued to the present day. None  considers the possibility of simply ignoring some scores as our equilibrium analysis suggests. 

\subsection{Related Work}
Concerns for fairness in standardized testing arose the instant they were introduced. In the US this dates to 1845, when Horace Mann deployed  standardized written exams as a replacement for the existing oral examination used for public school admission in Boston. Reese \cite{reese} writes:
\begin{quote}
    ``What transpired then still sounds eerily familiar: cheating scandals, poor performance by minority groups, the narrowing of curriculum, the public shaming of teachers, the appeal of more sophisticated measures of assessment, the superior scores in other nations, all amounting to a constant drumbeat about school failure.''
\end{quote}
For a recent survey of the fairness issue that standardized testing raises see Grodsky et al \cite{grodsky}, and see Hutchinson and Mitchell \cite{hutchinson} for a retrospective of the long history of thought on fairness in standardized testing, contextualized within the current literature on fairness in machine learning.

More recently, the widespread application of statistical techniques to high stakes decision making has led to a resurgence of interest in fairness in classification and prediction. False positive and false negative rates have once again been focal measures of unfairness across populations --- see e.g.  \cite{KMR16,Cho17,HPS16}. This literature is broad---here we provide a brief overview of the most relevant subset of this literature, which uses equilibrium analysis to make predictions about policy choices. Hu and Chen \cite{HC18} consider a two stage model of a labor market, and study interventions at the first (``internship'') stage that can lead to more equitable outcomes in equilibrium (See also Coate and Loury \cite{coate1993will} and Foster and Vohra \cite{FV92} for related models of self-confirming equilibria in labor markets from the economics literature). Liu et al \cite{liu2020disparate}  consider a model of the labor market with higher dimensional signals, and study equilibrium effects of ``subsidy'' interventions which can
lessen the cost of exerting effort. Kannan, Roth, and Ziani \cite{kannan2019downstream} study a two-stage pipeline in which colleges admit students from one of two populations based on a noisy signal about their type, and then commit to a grading policy which is used by a rational downstream employer to make hiring decisions. They study how policy decisions at the level of the college affect various measures of equity at the level of the employer's hiring decisions.  Three papers \cite{milli2019social,hu2019disparate,braverman2020role} study ``strategic classification'' problems in which individuals rationally manipulate their features in response to a deployed classifier in an attempt to optimize for their own outcome, and study the equilibrium effects of different populations having different costs for manipulation.  Immorlica, Ligett, and Ziani \cite{immorlica2019access} consider the related problem of \emph{population level signalling}, in which a third party (i.e. a highschool) is able to commit to a signalling scheme for an entire population (i.e. a grading policy), in a Bayesian Persuasion like model. They show that if one population of students is associated with a high school that is able to optimally signal, but another population of students is associated with a high school that naively signals, then against a Bayesian college, inequities arise in favor of the population associated with the optimal signalling technology --- and that counter-intuitively, the introduction of an unbiased standardized test (administered uniformly across all populations, unlike in our model) can sometimes \emph{exacerbate} this issue. Jung et al. \cite{jung2020fair} study an equilibrium model of criminal justice, in which individuals from populations that differ in their access to legal employment opportunities make rational decisions about whether or not to commit crime, as a function of the criminal justice policy --- and conclude that the crime-minimizing policy should commit to ignoring demographic information so as to equalize false positive and negative rates across populations.

%\ar{Need a preliminaries section here, formally laying out the model, giving notation to refer to the type and test outcomes of a student, etc.}

\Xomit{
\section{Benchmark: All Test Once}
\ar{What is the purpose of this benchmark? We should tell the reader} \ar{After reading the rest of the paper, I'm unsure of the value of this section. Perhaps just remove?}
\sk{agreed ... unless we compare it to the ReportAll and ReportMax solutions.}
In this section, we provide a benchmark in which each student tests once only. Here, the model is entirely symmetric across categories: category 2 students have no advantage over category 1 students, and the admission decision is made according to a single score. In this scenario, the students have no choices to make, so the model is not a game: the college must simply perform a Bayesian inference. After observing a score of $A$, the college updates its belief that the student is type $H$ from the prior $p$ to the posterior $$p_A = \frac{p \alpha}{p \alpha + (1-p)(1-\alpha)}.$$
After observing a score of $B$, the posterior is
$$p_B = \frac{p (1-\alpha)}{p (1-\alpha) + (1-p)\alpha}.$$
A student is accepted if the posterior is above $\frac{1}{2}$ and rejected otherwise. We can then characterize all possible cases as follows.
\begin{enumerate}
    \item If $p \leq 1-\alpha$, the College rejects regardless of the signal.
    \item If $p \in [1-\alpha, 1/2)$, the College accepts students with a high signal and rejects students with a low signal.
   % \item If $p \geq \alpha$, college accepts all regardless of the signal.
\end{enumerate}
}

\section{Reporting the Max Score}
It is clear that if $p$ is small enough, the College will be better off rejecting all students. This is an uninteresting outcome, and we will exclude it by assuming the $p$ is sufficiently large. Below, we identify a threshold ${\hat p}$ such that for $p < \hat{p}$ the {\em only} equilibrium under ``Report Max'' is to reject all students and we restrict attention to values of $p$ above that threshold, i.e. $p \in [\hat{p}, 0.5)$. 

We focus on a ``separating equilibrium'' in which the College accepts a student if the score is $A$ and rejects otherwise. This is the case when the test is effective in selecting students. For a narrow range of values in $[\hat{p}, 0.5)$ there is also an equilibrium where  all students are rejected. This is discussed in \textit{Section \ref{pool}}.

For convenience of exposition, in the following we write $\bar{x} \equiv 1-x$ for any variable $x \in [0,1]$.

\subsection{Separating Equilibrium}
Under the separating equilibrium, assuming it exists,  only the best score needs to be reported. Therefore a Category 2 student will take the test twice if needed to get a score of $A$. Hence, a Category 2 type $L$ student, denoted by $(2,L)$, will  report $B$ with probability $\alpha^2$. Similarly, a Category 2 type $H$ student, denoted by $(2,H)$, will report $A$ with probability $1- (1-\alpha)^2$. 
\textit{Table \ref{tab:prob_maxtest_sep}} summarizes the outcomes
\begin{table}[H]
    \centering
    \begin{tabular} {|c|c|c|c|c|} \hline
Best Score / Type & $\bm{(1,H)}$ & $\bm{(1,L)}$ & $\bm{(2,H)}$ & $\bm{(2,L)}$ \\ \hline
 $\bm{A}$ 
 & $\alpha \phi p$  &$\bar{\alpha} \phi \bar{p}$ & $(1-\bar{\alpha}^2)\bar{\phi}p$ & $(1-\alpha^2)\bar{\phi}\bar{p}$    \\ \hline
 $\bm{B}$
 & $\bar{\alpha} \phi p$ & $\alpha \phi \bar{p}$ & $\bar{\alpha}^2\bar{\phi}p$ &$\alpha^2\bar{\phi}\bar{p}$  \\ \hline
\textit{Total} & $\phi p$ & $\phi \bar{p}$ & $\bar{\phi}p$ & $\bar{\phi}\bar{p}$  \\ \hline
\end{tabular}
    \caption{\label{tab:prob_maxtest_sep}Probability Distribution with Best Test Reported [Separating Case]}
\end{table}

In the following theorem, we characterize exactly when this separating equilibrium exists under the ``Report Max'' policy:

\begin{theorem} Let  $\hat{p} =
%\frac{1 + \alpha \bar{\phi}}{\frac{1}{\bar{\alpha}} + 2\alpha \bar{\phi}}
\frac{1 + \alpha (1-\phi)}{\frac{1}{1-\alpha} + 2\alpha (1-\phi)} \in (1-\alpha, \frac{1}{2}).$
When $p \in [{\hat p}, 0.5]$, ``Report Max'' has a separating equilibrium in which the College accepts a student if the score is $A$ and rejects otherwise. Category 2 students take the test twice if they receive a score of $B$  on the first attempt. 
\end{theorem}
\begin{proof} From Table \ref{tab:prob_maxtest_sep}, we see that the probability that a student is of type $H$ given they report $A$ is
$$\frac{\alpha \phi p + (1-\bar{\alpha}^2)\bar{\phi}p}{\alpha \phi p  + \bar{\alpha} \phi \bar{p} +(1- \bar{\alpha}^2)\bar{\phi}p+ (1-\alpha^2)\bar{\phi}\bar{p} }.$$
If this exceeds 0.5, anyone who reports $A$ is admitted. This happens if
\begin{align*}
    \Big[\alpha \phi  + (1-\bar{\alpha}^2)\bar{\phi}\Big]p &\geq  \Big[\bar{\alpha} \phi + (1-\alpha^2)\bar{\phi}\Big](1-p)\\
   \Rightarrow p &\geq \hat{p}.
\end{align*}

From \textit{Table \ref{tab:prob_maxtest_sep}} we see that the probability of high type given a score of $B$ is
$$\frac{\bar{\alpha} \phi p + \bar{\alpha}^2\bar{\phi}p}{ \bar{\alpha} \phi p + \alpha \phi \bar{p}+\bar{\alpha}^2\bar{\phi}p+\alpha^2\bar{\phi}\bar{p}}.$$
If this is less than 0.5, then, anyone with a $B$ score is rejected. This happens if
\begin{align*}
   \Big[ \bar{\alpha} \phi  + \bar{\alpha}^2\bar{\phi} \Big]p
&\leq 
 \Big[\alpha \phi +\alpha^2\bar{\phi}\Big] (1-p),
\end{align*}
which holds true when $ p \leq 0.5 $.
\end{proof}

\remark{
Unsurprisingly, the more accurate the test, the larger the range of $p$'s for which there is a separating equilibrium. As the accuracy of the test approaches 1, one can have a separating equilibrium even if a vanishing fraction of students are of type H.}

\subsection{Reject All Equilibrium}\label{pool}
 We identify a threshold $\hat{\hat{p}}$ such that for $p \in [\hat{p}, \hat{\hat{p}}] \subset [\hat{p}, 0.5]$ there is also  a reject all equilibrium.  For $p \in (\hat{\hat{p}}, 0.5]$, the separating equilibrium is unique.
 \Xomit{
 \ar{Don't understand what is meany by ``only'' a reject all equilibrium. I thought above $\hat p$ we have the separating equilibrium. Maybe this is a typo?} 
 \mn{Agree with Aaron. The full picture should be:
 $[0,\hat{p})$ only reject all; $[\hat{p},\hat{\hat{p}})$ reject all + separating;
 $[\hat{\hat{p}}, \hat{\hat{p}}')$ only separating;
 $[\hat{\hat{p}}', \hat{p}']$ separating + accept all;
 $(\hat{p}', 1]$ accept all. And $ 0 < 1-\alpha <\hat{p}< \hat{\hat{p}}< \frac{1}{2} < \hat{\hat{p}}'<\alpha< \hat{p}' <1$.
 }.\sk{Yes, this makes sense... I made the change.}
}
 Let $f_L(B)$ be the probability that a $(2,L)$ student stops after taking one test with score $B$. Let $f_H(B)$ denote the probability that a  $(2,H)$ student stops after one test with score of $B$. \textit{Table \ref{tab:prob_maxtest}} summarizes the possible outcomes. 
\begin{table}[H]

\begin{tabular} {|c|c|c|c|c|} \hline
Best Score / Type & $\bm{(1,H)}$ & $\bm{(1,L)}$ & $\bm{(2,H)}$ & $\bm{(2,L)}$ \\ \hline
 $\bm{A}$ 
& $\alpha \phi p$  %1
&$\bar{\alpha} \phi \bar{p}$ %2
& $\alpha \bar{\phi}p [1 + \bar{\alpha} \bar{f}_H(B)]$ %3
& $\bar{\alpha} \bar{\phi} \bar{p} [1 + \alpha \bar{f}_L(B)]$ %4
    \\ \hline
 $\bm{B}$ 
& $\bar{\alpha} \phi p$ 
& $\alpha \phi \bar{p}$ 
& $\bar{\alpha}\bar{\phi}p [1 -\alpha \bar{f}_H(B)]$ &$\alpha \bar{\phi} \bar{p}[1 - \bar{\alpha} \bar{f}_L(B)]$ 
 \\ \hline
\textit{Total} & $\phi p$ & $\phi \bar{p}$ & $\bar{\phi}p$ & $\bar{\phi}\bar{p}$  \\ \hline
\end{tabular}
\caption{\label{tab:prob_maxtest}Probability Distribution with Best Test Reported [General Case]}
\end{table}

Suppose the College rejects all students regardless of reported scores. 
Since score doesn't matter for the admission decision, any $f_L(B) \in [0,1], f_H(B) \in [0,1]$ are best responses for Category 2 students in this case.  As the College rejects students with a score of $A$, the following must hold: 
    \[
\frac{\alpha \phi p + \alpha \bar{\phi}p [1 + \bar{\alpha}\bar{f}_H(B)] }{\alpha \phi p + \bar{\alpha} \phi \bar{p} +  \alpha \bar{\phi} p [1 + \bar{ \alpha}\bar{f}_H(B)] + \bar{\alpha} \bar{\phi} \bar{p} [1 + \alpha \bar{f}_L(B)]} \leq \frac{1}{2} 
\] 

\begin{equation}\label{xL1}
\Rightarrow \quad \bar{f}_L(B) \geq \frac{p - \bar{\alpha} + \alpha \bar{\alpha} \bar{\phi} p \bar{f}_H(B)}{\alpha \bar{\alpha} \bar{\phi} \bar{p} }.
\end{equation}
 As the College rejects students with signal $B$, the following must be true: 

\[
\frac{\bar{\alpha} \phi p + \bar{\alpha}\bar{\phi}p [1 -\alpha\bar{f}_H(B))]}{\bar{\alpha}\phi p + \alpha \phi \bar{p} + \bar{\alpha}\bar{\phi}p [1 -\alpha\bar{f}_H(B)] + \alpha\bar{\phi}\bar{p}[1 - \bar{\alpha} \bar{f}_L(B)]}
\leq \frac{1}{2} 
\] 

\begin{equation}\label{xL2}
\Rightarrow \quad \bar{f}_L(B) \leq
 \frac{-p + \alpha + \alpha \bar{\alpha} \bar{\phi} p \bar{f}_H(B)}{\alpha \bar{\alpha} \bar{\phi} \bar{p} }.
\end{equation}
Combining (\ref{xL1}) and (\ref{xL2}) we deduce that

\begin{align*}
\frac{-p + \alpha + \alpha \bar{\alpha} \bar{\phi} p \bar{f}_H(B)}{\alpha \bar{\alpha} \bar{\phi} \bar{p} } \geq \frac{p - \bar{\alpha} + \alpha \bar{\alpha} \bar{\phi} p \bar{f}_H(B)}{\alpha \bar{\alpha} \bar{\phi} \bar{p} } 
&\quad \Rightarrow \quad  
p \leq \frac{1}{2},\\    
\frac{-p + \alpha + \alpha \bar{\alpha} \bar{\phi} p \bar{f}_H(B)}{\alpha \bar{\alpha} \bar{\phi} \bar{p} } \geq \; 0 
&\quad \Leftarrow \quad  
p \leq \frac{1}{2},\\
\frac{p - \bar{\alpha} + \alpha \bar{\alpha} \bar{\phi} p \bar{f}_H(B)}{\alpha \bar{\alpha} \bar{\phi} \bar{p} } \leq \; 1 
&\quad \Rightarrow \quad  
p \leq \frac{\alpha \bar{\alpha} \bar{\phi} + \bar{\alpha}}{\alpha \bar{\alpha} \bar{\phi} + 1}, \bar{f}_H(B) \leq  \frac{\alpha \bar{\alpha} \bar{\phi} \bar{p}+ \bar{\alpha} -p}{\alpha \bar{\alpha} \bar{\phi} p}.
\end{align*}
 Therefore, for $p \leq \hat{\hat{p}}$ where $\hat{\hat{p}} = \min\{\frac{1}{2}, \frac{\alpha \bar{\alpha} \bar{\phi} + \bar{\alpha}}{\alpha \bar{\alpha} \bar{\phi} + 1}\} \in (\hat{p}, \frac{1}{2}] $ , there exists an equilibrium in which the College rejects all students. In such an equilibrium, $(2,L)$ take a second test with probability $\bar{f}_L(B)$ when the first score is $B$, $(2,H)$ students take a second test with probability $\bar{f}_H(B)$ when the first score is $B$ and we have $\bar{f}_H(B) \in [0, \frac{\alpha \bar{\alpha} \bar{\phi} \bar{p}+ \bar{\alpha} -p}{\alpha \bar{\alpha} \bar{\phi} p}], \bar{f}_L(B) \in [\frac{p - \bar{\alpha} + \alpha \bar{\alpha} \bar{\phi} p \bar{f}_H(B)}{\alpha \bar{\alpha} \bar{\phi} \bar{p} },\frac{-p + \alpha + \alpha \bar{\alpha} \bar{\phi} p \bar{f}_H(B)}{\alpha \bar{\alpha} \bar{\phi} \bar{p} }]$.

\section{Reporting All Scores}\label{sec: report all}
Here we characterize the equilibrium outcome under ``Report All''. 
%There are two equilibria, but they generate exactly the same outcome. They differ only in whether category 2 students take the test once or twice.

Let $u_B$ be 1 if the College accepts a student reporting a single $B$ score and zero otherwise. Similarly define
$u_A, u_{BA}, u_{AB}, u_{BB}, u_{AA}$ to be the indicators of whether the College accepts a student reporting the corresponding sequence of exam scores. We proceed by examining whether various combinations of values for these $u$ variables can be supported in equilibrium. It might seem that some combinations could be eliminated immediately---but things are not so simple.

Sometimes our intuition will be confirmed. For example, it may be obvious that we should have $u_B=u_{BB} =0$. Why would the College accept a student who only reports $B$s? If $p$ were large, say close to 1, then it would. But in our case, $p < 0.5$ and our analysis will confirm that for such values of $p$, we have $u_B=u_{BB} =0$ in equilibrium.

Now a more counter-intuitive case: Should $u_{AB} = u_{BA}$? After all, why should the timing of a $B$ score matter? If these test scores were exogenously given to us, then by symmetry, they would induce the same posterior belief on a student's type. But our analysis  shows that in equilibrium, there is an important distinction  because of the incentives induced in equilibrium for  students to take a second test --- and this in turn affects the inferences the College makes. 

\Xomit{We now introduce notation to encode the strategies of the students and their expected payoffs. 

In Sections \ref{bestH} and \ref{bestL} we determine the best responses of the students to the possible choices of $u_B, u_A, u_{BA}, u_{AB}, u_{BB}, u_{AA}$. Section \ref{equilibrium} characterizes the unique equilibrium.
\ar{I thought there were 2 equilibria. Here we mean the unique equilibrium strategy for the College?}
\mn{There are a series of equilibria (depending on College's beliefs off the equilibrium path) and thus multiple equilibrium strategies for the College accordingly. Only the equilibrium outcome is unique in the range of interest.}
Specifically, $u_A= u_{AB}=u_{AA}=1$ and $u_B = u_{BA}=u_{BB} = 0$\mn{We can only say in equilibrium $u_A=1, u_B =u_{BA}=u_{BB}= 0$, while $u_{AB}, u_{AA} \in \{0,1\}$}.}

We introduce the following variables to track the actions of the category 2 students as a function of their first test score: 
\begin{enumerate}
    \item $f_L(A)$: the probability that a $(2,L)$ students stops after one test with a score of $A$.
    \item $f_L(B)$: the probability that a $(2,L)$ student stops after one test with a score of $B$.
    \item $f_H(A)$: the probability that a $(2,H)$ students stops after one test with a score of $A$.
    \item $f_H(B)$: the probability that a $(2,H)$ students stops after one test with a score of $B$.
\end{enumerate}

%\sk{Most of the other quantities in this paper are probabilities, not fractions... could we define $f_L(B)$ etc., to be the probability that a student of type $(2,L)$ stops after one score of $B$,etc.? This shouldn't make a difference in the calculations, but would be consistent with the rest of the paper. It also makes more sense to use probabilities when computing an individual student's payoff.}

We remark that we can interpret these probabilities as the \emph{fraction} of a student population that takes the corresponding action, so we do not have to imagine that individual students randomize. 

Fixing student strategies, the probabilities of each  testing outcome are displayed in \textit{Table \ref{tab:prob_alltests}}. 

\begin{table}[H]
\begin{tabular} {|c|c|c|c|c|} \hline
Score / Type & $\bm{(1,H)}$ & $\bm{(1,L)}$ & $\bm{(2,H)}$ & $\bm{(2,L)}$ \\ \hline
 $\bm{A}$ 
 & $\alpha \phi p$  &$\bar{\alpha} \phi \bar{p}$ & $\alpha\bar{\phi}pf_H(A)$ & $\bar{\alpha}\bar{\phi}\bar{p}f_L(A)$     \\ \hline
 $\bm{B}$ 
 & $\bar{\alpha} \phi p$ & $\alpha \phi \bar{p}$ & $\bar{\alpha}\bar{\phi}pf_H(B)$ &$\alpha\bar{\phi}\bar{p}f_L(B)$  \\ \hline
 $\bm{AA}$ 
 &0 &0 &$\alpha^2\bar{\phi}p\bar{f}_H(A)$ &$\bar{\alpha}^2\bar{\phi}\bar{p}\bar{f}_L(A)$   \\ \hline
$\bm{AB}$ 
 & 0&0 &$\alpha\bar{\alpha}\bar{\phi}p\bar{f}_H(A)$ & $\alpha\bar{\alpha}\bar{\phi}\bar{p}\bar{f}_L(A)$  \\ \hline
 $\bm{BA}$ 
 & 0&0 &$\alpha\bar{\alpha}\bar{\phi}p\bar{f}_H(B)$ & $\alpha\bar{\alpha}\bar{\phi}\bar{p}\bar{f}_L(B)$  \\ \hline
 $\bm{BB}$ 
 &0 &0 & $\bar{\alpha}^2\bar{\phi}p\bar{f}_H(B)$& $\alpha^2\bar{\phi}\bar{p}\bar{f}_L(B)$  \\ \hline
 \textit{Total }
 & $\phi p$ & $\phi \bar{p}$ & $\bar{\phi}p$ & $\bar{\phi}\bar{p}$  \\ \hline
\end{tabular}
\caption{\label{tab:prob_alltests}Distribution of Testing Outcomes}
\end{table}

The following theorem characterizes equilibrium outcomes under ``Report All" within the parameter range of interest.

\begin{theorem}[First-Score Equilibrium]
\label{th:all_2_eqm}
%Let  $\hat{p} =
%\frac{1 + \alpha \bar{\phi}}{\frac{1}{\bar{\alpha}} + 2\alpha \bar{\phi}}
%\frac{1 + \alpha (1-\phi)}{\frac{1}{1-\alpha} + 2\alpha (1-\phi)} \in (1-\alpha, \frac{1}{2}).$
For $p \in (1-\alpha, 0.5)$, the unique equilibrium outcome for ``Report All" is that a student is admitted if their first (or only) score is $A$ and rejected otherwise.
% , i.e., $u_A=u_{AB}=u_{AA}=1$ and $u_B=u_{BA}=u_{BB}=0$.
\end{theorem}

\begin{remark}
For clarity, we prove \textit{Theorem \ref{th:all_2_eqm}} here, but we remark that it can also be derived as a  corollary of the more general \textit{Theorem \ref{th: all_eqm}} for the case in which Category 2 students may take $k \geq 2$ tests. We prove  this more general theorem in \textit{Section \ref{sec: report all_multi}}.
\end{remark}

\begin{proof}
Let  the College's posterior belief that a student is of type $H$ after observing some (sequence of) reported scores $s\in \{B,A,BB, BA, AB,AA\}$ be $p_s \in [0,1]$. The College's expected payoff from admitting a student with score $s$ is $1 \cdot p_s +  (-1) \cdot (1-p_s) = 2p_s - 1$. Hence, the optimal admission policy for college is:
$$u_s = \begin{cases}
    1  &\quad \text{ if } p_s > \frac{1}{2},\\
    0 \text { or } 1 &\quad \text{ if } p_s= \frac{1}{2},\\
    0 &\quad \text{ if } p_s < \frac{1}{2}.
    \end{cases}$$ 
    From \textit{Table \ref{tab:prob_alltests}}, we can compute:
\begin{align}
    p_A &= \frac{\phi p \alpha + \bar{\phi} p \alpha f_H(A)}{\phi(p \alpha + \bar{p}\bar{\alpha}) + \bar{\phi}[ p \alpha f_H(A) + \bar{p}\bar{\alpha} f_L(A)]} \in (0, 1), \label{pA}\\
     p_B &= \frac{\phi p \bar{\alpha} + \bar{\phi} p \bar{\alpha} f_H(B)}{\phi(p \bar{\alpha} + \bar{p}\alpha) + \bar{\phi}[ p \bar{\alpha} f_H(B) + \bar{p}\alpha f_L(B)]} \in (0, 1), \label{pB}\\
     p_{AA} &= \frac{p \alpha^2 \bar{f}_H(A)}{p \alpha^2 \bar{f}_H(A) + \bar{p}\bar{\alpha}^2\bar{f}_L(A)} \quad \text{if $\, p \alpha^2 \bar{f}_H(A) + \bar{p}\bar{\alpha}^2\bar{f}_L(A)>0$}, \label{pAA}\\
    p_{AB} &= \frac{p \bar{f}_H(A)}{p \bar{f}_H(A) + \bar{p}\bar{f}_L(A)} \quad \quad \quad \, \text{if $\, p \bar{f}_H(A) + \bar{p}\bar{f}_L(A)>0$},\label{pAB}\\
    p_{BA} &= \frac{p \bar{f}_H(B)}{p \bar{f}_H(B) + \bar{p}\bar{f}_L(B)} \quad \quad \quad \, \text{if $\, p \bar{f}_H(B) + \bar{p}\bar{f}_L(B)>0$},\label{pBA}\\
    p_{BB} &= \frac{p \bar{\alpha}^2 \bar{f}_H(B)}{p \bar{\alpha}^2 \bar{f}_H(B) + \bar{p}\alpha^2\bar{f}_L(B)} \quad \text{if $\, p \bar{\alpha}^2 \bar{f}_H(B) + \bar{p}\alpha^2\bar{f}_L(B)>0$}. \label{pBB}
\end{align}

Note that $p_{sB}<p_{sA}$ if $f_L(s) <1$ and $f_H(s) <1$ for any $s\in \{A,B\}$. Also, $p_{AB} \neq p_{BA}$ if any of the following scenarios appear in equilibrium : (i) all students test once when their first score is $A$; or (ii) all students test once when their first score is $B$; or (iii) $\frac{\bar{f}_H(A)}{\bar{f}_L(A)} \neq \frac{\bar{f}_H(B)}{\bar{f}_L(B)}$. 

Next we will prove by contradiction that for all $p \in (\bar{\alpha}, 0.5)$, the College admits students submitting a single score $A$ and rejects students submitting $B, BA$, or  $BB$ in equilibrium. Thus, in any equilibrium, the admission outcome depends solely on the first score. Here we split the discussion according to whether the student's first score is $A$ or $B$.\\
\\
\textbf{Case 1: The First Score is $\mathbf{A}$}\par
Suppose, for contradiction, the College rejects students who submit a single score of $A$. Then we have the following two cases: 
\begin{enumerate}
\item If the College rejects all scores starting in $A$, namely $A, AA, AB$, then we need $\max\{p_A, p_{AA}, p_{AB}\} \leq \frac{1}{2}$ to rationalize this admission rule. Hence by \textit{the law of total probability}:
\small
\begin{align*}
p\alpha  &= 
p_A \{\phi(p \alpha + \bar{p}\bar{\alpha}) + \bar{\phi}[ p \alpha f_H(A) + \bar{p}\bar{\alpha} f_L(A)]\} 
+ p_{AA} \bar{\phi}[p \alpha^2 \bar{f}_H(A) + \bar{p}\bar{\alpha}^2\bar{f}_L(A)] + p_{AB} \bar{\phi}\alpha \bar{\alpha}[p \bar{f}_H(A) + \bar{p}\bar{f}_L(A)] \\
&\leq \frac{1}{2} \{\phi(p \alpha + \bar{p}\bar{\alpha}) + \bar{\phi}[ p \alpha f_H(A) + \bar{p}\bar{\alpha} f_L(A)] + \bar{\phi}[p \alpha^2 \bar{f}_H(A) + \bar{p}\bar{\alpha}^2\bar{f}_L(A)] + \bar{\phi}\alpha \bar{\alpha}[p \bar{f}_H(A) + \bar{p}\bar{f}_L(A)]\}\\
& = \frac{1}{2} (p\alpha + \bar{p} \bar{\alpha}). 
\end{align*}
\normalsize
This implies $\frac{p\alpha}{p\alpha + \bar{p} \bar{\alpha}} \leq \frac{1}{2}$, i.e., $p \leq \bar{\alpha}$, a contradiction.

\item If the College admits students reporting either $AA$ or $AB$, then all Category 2 students take the test a second time after obtaining a first score of $A$ (i.e., $f_H(A) = f_L(A) = 0$). In this case, the only students reporting a single score of $A$ are from Category 1, and so by \textit{Equation \ref{pA}} we have $p_{A} = \frac{p\alpha}{p\alpha + \bar{p} \bar{\alpha}} \leq \frac{1}{2}$.  This implies $p \leq \bar{\alpha}$, a contradiction again.
\end{enumerate}
Therefore, as long as $p \in (\bar{\alpha}, 0.5)$, a single score $A$ yields admission. \\
\\
\textbf{Case 2: The First Score is $\mathbf{B}$}\par
Suppose, for a contradiction, the College admits students reporting any of $B, BA$ or $BB$. Then we have the following three cases: 
\begin{enumerate}
\item If the College admits all scores starting in $B$, namely $B, BA, BB$, then we need $\min\{p_B, p_{BA}, p_{BB}\} \geq \frac{1}{2}$ to rationalize this admission rule. Hence by \textit{the law of total probability}:
\small
\begin{align*}
p\bar{\alpha}  &= 
p_B \{\phi(p \bar{\alpha} + \bar{p}\alpha) + \bar{\phi}[ p \bar{\alpha} f_H(B) + \bar{p}\alpha f_L(B)]\} 
+ p_{BA} \alpha \bar{\alpha} \bar{\phi}[p \bar{f}_H(B) + \bar{p}\bar{f}_L(B)] 
+ p_{BB} \bar{\phi}[p \bar{\alpha}^2 \bar{f}_H(B) + \bar{p}\alpha^2\bar{f}_L(B)] \\
&\geq \frac{1}{2} \{\phi(p \bar{\alpha} + \bar{p}\alpha) + \bar{\phi}[ p \bar{\alpha} f_H(B) + \bar{p}\alpha f_L(B)] + \alpha \bar{\alpha} \bar{\phi}[p \bar{f}_H(B) + \bar{p}\bar{f}_L(B)]  + \bar{\phi}[p \bar{\alpha}^2 \bar{f}_H(B) + \bar{p}\alpha^2\bar{f}_L(B)] \}\\
& = \frac{1}{2} (p\bar{\alpha} + \bar{p} \alpha). 
\end{align*}
\normalsize
This implies $\frac{p\bar{\alpha}}{p\bar{\alpha} + \bar{p} \alpha} \geq \frac{1}{2}$, i.e., $p \geq \alpha > 0.5$, a contradiction.

\item If the College admits students reporting $B$ but rejects students reporting either $BA$ or $BB$, then all Category 2 students take the test once if their first score is $B$---since they won't run the risk of being rejected due to the second test result. Thus by \textit{Equation \ref{pB}}, $p_{B} = \frac{p\bar{\alpha}}{p\bar{\alpha} + \bar{p} \alpha} \geq \frac{1}{2}$. This implies $p \geq \alpha > 0.5$, a contradiction again.

\item If the College rejects students reporting a score of $B$ but admits students reporting either a score of $BA$ or $BB$, then all Category 2 students take the test twice if their first score is $B$. In this case, by \textit{Equation \ref{pBA}} and \textit{Equation \ref{pBB}}, $p_{BB} < p_{BA} = p <\frac{1}{2}$, a contradiction to the assumption that the College admits either $BA$ or $BB$.
\end{enumerate}
Therefore, as long as $p \in (\bar{\alpha}, 0.5)$, a score beginning in $B$ yields rejection. This completes the proof. 
\end{proof}

In what follows, we refer to the set of equilibria in which the admissions outcome is made solely as a function of the first score as the \emph{first-score equilibrium}.

\section{Comparisons}
In this section, we compare the separating equilibrium under ``Report Max'' with the (unique) first-score equilibrium under ``Report All'', both from the perspective of student equity, and from the perspective of the College's objective. 

We recall that the \emph{false negative} rate corresponds to the proportion of High type students who are rejected, and the \emph{false positive} rate corresponds to the proportion of Low type students who are accepted. It is straightforward to compare the false positive and false negative rates of the separating equilibrium under ``Report Max''  with the first-score equilibrium under  ``Report All'. We summarize the results in the table below.
\begin{table}[H]
    \centering
    \begin{tabular} {|c|c|c|c|} \hline
 FN &  $\bm{(1,H)}$  & $\bm{(2,H)}$  \\ \hline
\textbf{Max} & $1-\alpha$  &  $(1-\alpha)^2$     \\ \hline
\textbf{All}  & $1-\alpha$ &   $1-\alpha$   \\ \hline
\end{tabular}
\quad
    \begin{tabular} {|c|c|c|c|} \hline
FP & $\bm{(1,L)}$  & $\bm{(2,L)}$  \\ \hline
\textbf{Max} & $1-\alpha$  &  $1-\alpha^2$     \\ \hline
\textbf{All} & $1-\alpha$ &   $1-\alpha$   \\ \hline
\end{tabular}
    \caption{\label{tab:FNFP} False Negative (left) and False Positive (right) Rates Across Categories of Students}
\end{table}

%\begin{table}[H]
%    \centering
%    \begin{tabular} {|c|c|c|c|} \hline
% & $\bm{(1,L)}$  & $\bm{(2,L)}$  \\ \hline
%\textbf{Max} & $1-\alpha$  &  $1-\alpha^2$     \\ \hline
%\textbf{All} & $1-\alpha$ &   $1-\alpha$   \\ \hline
%\end{tabular}
%    \caption{\label{tab:FP} False Positive Rates Across Categories of Students}
%\end{table}

We observe that the ``Report Max'' equilibrium favors the advantaged (Category 2) students, in that for each type $L,H$, their probability of admissions is strictly higher compared to the disadvantaged (Category 1) students of the same type. This manifests itself as both a higher false positive rate and a lower false negative rate, compared to the ``Report All'' equilibrium. In contrast, because the College in the ``Report All'' equilibrium makes decisions only as a function of the first test score, the probability of admissions conditional on type is identical across advantaged and disadvantaged students. This manifests itself as an identical false positive and false negative rate across categories. We can conclude that from the perspective of \emph{equity} across advantaged and disadvantaged students, the ``Report All" policy is preferable to the ``Report Max" policy. 

Next, we compare the \emph{positive predictive value} of the equilibrium outcomes for both policies --- i.e. the probability, in equilibrium, that a student is a High type, conditional on receiving admissions to the college. Higher positive predictive values will correspond to admissions outcomes with a higher proportion of High type students among the admitted class, and are hence desirable. We find that the positive predictive value is strictly higher for the ``Report All'' policy:

 \begin{theorem}\label{th:PPV_2}
    For any $\alpha \in (\frac{1}{2}, 1)$, $p < \frac{1}{2}$, the positive predictive value of the ``Report All'' policy strictly exceeds that of the ``Report Max'' policy in any nontrivial equilibrium.
    \end{theorem}
    \begin{proof}
    %For the sake of comparison, we need $p \in [\hat{p},\frac{1}{2})$ such that equilibrium is nontrivial under both admission policies\footnote{In the range $[0, \hat{p})$, there is only trivial equilibrium of rejecting all as shown in \textit{Section \ref{pool}} under ``Report Max". So the statement in \textit{Theorem \ref{th:PPV_2}} holds vacuously true.}. 
    The positive predictive value of ``Report Max'' is 
$$
  \frac{\alpha \phi p + (1-\bar{\alpha}^2)\bar{\phi}p}{\alpha \phi p + \bar{\alpha} \phi \bar{p} + (1-\bar{\alpha}^2)\bar{\phi}p + (1-\alpha^2)\bar{\phi}\bar{p}} = \frac{(1 + \bar{\alpha}\bar{\phi})\alpha p}{ \alpha p + \bar{\alpha} \bar{p} + \alpha \bar{\alpha}\bar{\phi}}.$$

The positive predictive value of ``Report All'' coincides with the policy of only looking at the first score, so it has value 
$$\frac{\alpha p}{\alpha p + \bar{\alpha}\bar{p}}.$$
It is straightforward to verify that the first expression is strictly smaller than the second if $\alpha \in (0.5, 1)$.
\end{proof}

 For completeness, we compare the \emph{negative predictive value} of the College's admissions rule used under both policies. The negative predictive value is the probability, in equilibrium, that a student is a Low type, conditional on being rejected from the college. 
 
 \begin{theorem}
    For any $\alpha \in (\frac{1}{2}, 1)$ and $p < 0.5$,  the negative predictive value of the ``Report All'' policy is strictly smaller than that of the ``Report Max'' policy in any nontrivial equilibrium.
    \end{theorem}
    \begin{proof}
%For the sake of comparison, we need $p \in [\hat{p},\frac{1}{2})$ such that equilibrium is nontrivial under both admission policies. 
The negative predictive value of ``Report Max'' is
$$\frac{\alpha \phi \bar{p} + \alpha^2\bar{\phi}\bar{p}}{\bar{\alpha}\phi p + \alpha \phi \bar{p} + \bar{\alpha}^2\bar{\phi}p + \alpha^2\bar{\phi}\bar{p}}.$$
The negative predictive value of ``Report All'' is
$$\frac{\bar{p}\alpha}{p\bar{\alpha} + \bar{p}\alpha}.$$
It is straightforward to verify that the first is larger than the second for $\alpha > \frac{1}{2}$. 
\Xomit{
$$\frac{ \phi + \alpha(1-\phi)}{(1-\alpha)\phi p + \alpha \phi (1-p) + (1-\alpha)^2(1-\phi)p + \alpha^2(1-\phi)(1-p)} >
\frac{1}{p(1-\alpha) + (1-p)\alpha}$$
$$\alpha(1-\phi)>
\frac{(1-\alpha)^2(1-\phi)p + \alpha^2(1-\phi)(1-p)}{p(1-\alpha) + (1-p)\alpha}
$$
$$
\alpha >
\frac{(1-\alpha)^2p + \alpha^2(1-p)}{p(1-\alpha) + (1-p)\alpha}=
\frac{p(1-\alpha)}{p(1-\alpha) + (1-p)\alpha}(1-\alpha)+
\frac{(1-p)\alpha}{p(1-\alpha) + (1-p)\alpha}\alpha.
$$
}
\end{proof}

Finally, we compare the college's utility (i.e. its \emph{classification accuracy} for the task of distinguishing High and Low type students) in the equilibrium outcomes for both policies. We find that the College has higher utility under the ``Report All'' policy, demonstrating that not only is ``Report All'' better from the perspective of student equity, but it is strictly better from the perspective of the College as well. 

\begin{theorem}
   For any $\alpha \in (\frac{1}{2}, 1)$ and $p < 0.5$, the College's expected payoff per student under ``Report All'' exceeds that under ``Report Max'' in equilibrium.
\end{theorem}

\begin{proof}
%For the sake of comparison, we need $p \in [\hat{p},\frac{1}{2})$ such that equilibrium is nontrivial under both admission policies. 
The expected payoff per student under ``Report All'' is $\alpha p - \bar{\alpha}\bar{p}$. The expected payoff per student under ``Report Max'' is
$\phi(\alpha p - \bar{\alpha}\bar{p}) + \bar{\phi}[(1-\bar{\alpha}^2)p - (1-\alpha^2)\bar{p}].$ Hence, the difference between the College's expected payoff under these two schemes is:
\begin{align*}
&\alpha p - \bar{\alpha}\bar{p} - \{\phi(\alpha p - \bar{\alpha}\bar{p} + \bar{\phi}[(1-\bar{\alpha}^2)p - (1-\alpha^2)\bar{p}]\}   \\
=\quad &\bar{\phi}[\alpha p - \bar{\alpha}\bar{p} -  (1-\bar{\alpha}^2)p + (1-\alpha^2)\bar{p}]\\
=\quad &\bar{\phi} \alpha \bar{\alpha} (1-2p) > 0 
\end{align*}
where the last inequality holds when $\alpha\in (0,1), p < \frac{1}{2}$.
%$$\alpha p - (1-\alpha)(1-p) \geq
%    \phi[\alpha p - (1-\alpha)(1-p)] + (1-\phi)[(1-(1-\alpha)^2)p - (1-\alpha^2)(1-p)]
%$$
% $$\alpha p - (1-\alpha)(1-p) \geq
%    (1-(1-\alpha)^2)p - (1-\alpha^2)(1-p)
%$$ 
%$$[\alpha - 1 + (1-\alpha)^2]p \geq [(1-\alpha) - (1-\alpha^2)](1-p)$$
%$$(\alpha^2 - \alpha)p \geq (\alpha^2 - \alpha)(1-p)$$
%$$p \leq 0.5$$
\end{proof}
Note that for all the quantities we have considered, equality occurs when $\alpha = 1$ (i.e., when the test is a perfect noiseless signal of student type). In this case, the College is able to admit exactly the High type students and reject exactly the Low type students under both reporting schemes, and  there is no advantage of one over the other either in terms of equity or accuracy. As the accuracy of the test $\alpha$ decreases, all the disparities we have measured grow monotonically.

%%%%%%%%%%%%%%%%%%%%%%%%%%%%%%%%%%%%%%%%%%%%%%%%%%%%%%%%%%%%%%%%%%%%
\section{The General Case: More Tests, Wider Parameter Ranges}\label{moretwo}
Thus far we have assumed that the advantaged students (Category 2) are able to take up to two tests, and that High type students are rare --- i.e. that $p < \frac{1}{2}$. In this section, we generalize our results to a broader setting. Now, Category $2$ students can adaptively choose to take the test up to $k$ times for an arbitrary $k \geq 2$, and we consider the full parameter range $p \in (0,1)$ . As one might imagine, the set of possible equilibria under ``Report All'' expands. Nevertheless, we identify a wide range of parameters for which basing admission on the first reported test score only  continues to be an equilibrium outcome. This equilibrium outcome maintains the advantages of the ``Report All'' equilibrium outcome we studied in the special case of $k = 2$.  Moreover, we prove that \emph{all} equilibria of the ``Report All'' policy dominate the equilibrium of the ``Report Max'' policy in terms of both equity and accuracy: the College continues to strictly prefer outcomes under the ``Report All'' policy for all possible equilibria, and similarly, all such equilibria have smaller false positive and false negative gaps between Categories of students, as compared to the Report Max equilibrium. 
%So, we confine ourselves to identifying a subset of $[\hat{p}, 0.5)$ in which the category 2 students take the test at most twice.

\subsection{Report Max}
We start by characterizing the equilibrium outcome under ``Report Max''. 
As before we focus on a separating equilibrium in which the College accepts a student if the reported score is $A$ and rejects otherwise. Under this separating equilibrium, assuming it exists,  only the best score will be reported. Therefore Category 2 students will take the test as often as needed to get an $A$ score. Therefore a $(2,L)$ student will  report $B$ with probability $\alpha^k$, and a $(2,H)$ student will report $A$ with probability $1- (1-\alpha)^k$. 

\begin{theorem}\label{th: max_multi_eqm}
Let  $\hat{p}_{k} =
\frac{\phi (1-\alpha) + (1-\phi) (1-\alpha^k)}{\phi + (1-\phi) [2 -\alpha^k -(1-\alpha)^k]}$ and $\hat{p}'_{k} =
\frac{\phi \alpha + (1-\phi) \alpha^k}{\phi + (1-\phi) [\alpha^k +(1-\alpha)^k]}$ for any $k\geq 2$. ``Report Max'' has a nontrivial (separating) equilibrium if and only if 
$$p \in [\hat{p}_{k}, \hat{p}'_{k}].$$
Of course, the nontrivial equilibrium is unique: the College accepts a student if the reported score is $A$ and rejects otherwise. Category 2 students take the exam as many times as they need to get an $A$ score (up to $k$ times).
\end{theorem}

Note that the parameter range in which a separating equilibrium exists is always nontrivial since $\hat{p}'_{k} > \alpha > \frac{1}{2}> \hat{p}_{k} > \bar{\alpha}$ for all $k \geq 2$. Observe also that $\hat{p}_{k}$ strictly increases in $k$, which captures the intuition that under ``Report Max", increased testing makes a report of $A$ less indicative of a high type.

\begin{proof} We can calculate that the probability that a student is of type $H$ given they report $A$ is
$$\frac{\alpha \phi p + (1-\bar{\alpha}^k)\bar{\phi} p}{\alpha \phi p  + \bar{\alpha} \phi \bar{p} +(1- \bar{\alpha}^k)\bar{\phi}p+ (1-\alpha^k)\bar{\phi}\bar{p} }.$$
If this exceeds 0.5, then anyone who reports $A$ is admitted by the College. This happens exactly when:
\begin{align*}
    \Big[\alpha \phi  + (1-\bar{\alpha}^k)\bar{\phi} \Big]p &\geq  \Big[\bar{\alpha} \phi + (1-\alpha^k)\bar{\phi} \Big](1-p)\\
   \Rightarrow p &\geq \hat{p}_{k}
\end{align*}

The probability of being a high type given a report of $B$ is
$$\frac{\bar{\alpha} \phi p + \bar{\alpha}^k\bar{\phi}p}{ \bar{\alpha} \phi p + \alpha \phi \bar{p}+\bar{\alpha}^k \bar{\phi} p+\alpha^k \bar{\phi} \bar{p}}.$$
If this is less than 0.5, then, anyone with an $B$ score is rejected. This happens exactly when
\begin{align*}
   \Big[ \bar{\alpha} \phi  + \bar{\alpha}^k\bar{\phi} \Big]p
&\leq 
 \Big[\alpha \phi+\alpha^k \bar{\phi}\Big] (1-p)\\
 \Rightarrow p &\leq \hat{p}'_{k}.
\end{align*}
\end{proof}

\subsection{Report All}\label{sec: report all_multi}
We now turn to the ``Report All'' policy. We find that within a  wide range of parameters ($p \in [1-\alpha,\alpha]$) there still \emph{exists} an equilibrium of the sort that we had in the $k = 2$ case --- namely, in which the College makes decisions only based on the first test score, and hence which treats both populations equally. We also characterize the parameter range in which other non-trivial equilibria exist. 
\begin{theorem}\label{th: all_eqm}
   Let $p^*_k =  \frac{\bar{\alpha}^{k-2}}{\alpha^{k-2} + \bar{\alpha}^{k-2}}$ for $k\geq 2$. Under ``Report All", for any $\alpha \in (\frac{1}{2},1]$,
   \begin{enumerate}
       \item There exists an equilibrium in which the admission outcome depends solely on the first score if and only if $$p \in [1-\alpha, \alpha].$$
       
       \item There exists an equilibrium in which the admission outcome depends on more than the first score if and only if $$p \in [p^*_{k+2}, 1-\alpha] \cup [p^*_k, \alpha].$$
       
       \item For any $p \in (1-\alpha, \alpha)$, a reported (single) score of $A$ yields admission and a reported score sequence that consists entirely of $B$ scores yields rejection in \emph{all} equilibria. 
   \end{enumerate}
\end{theorem}

\begin{proof}
We write $p_s \in [0,1]$ to denote the college's posterior belief that the student is of High type after observing a reported score $s\in \cup_{i=1}^{k} \{A,B\}^i$.\footnote{$\cup_{i=1}^{k} \{A,B\}^i = \{A, B, AA, AB, BA, BB, \cdots, \underbrace{A \hdots A}_{k}, \cdots, \underbrace{B \hdots B}_{k}\}$ denotes the set of all possible score sequences that can result from taking the test up to $k$ times.} For the analysis, we will also be interested in  the proportion of high type students among all students whose reported scores \emph{start} with $s$ (and might have an arbitrary continuation). We denote this by $p_{s*} \in [0,1]$ with the convention that $p_{s*} = 0$ if no score starts in $s$ in equilibrium. 
 For any $s$ the optimal admissions policy for the college remains:
$$u_s = \begin{cases}
    1  &\quad \text{ if } p_s > \frac{1}{2},\\
    0 \text { or } 1 &\quad \text{ if } p_s= \frac{1}{2},\\
    0 &\quad \text{ if } p_s < \frac{1}{2}.
    \end{cases}$$ 
    
It will be useful to group testing outcomes by whether they consist of a single score, or more than one score. This distinction is important because Category 2 students can  pool with Category 1 students only if they take the test only once. Just as in \textit{Table \ref{tab:prob_alltests}}, we display the probabilities of testing outcomes in \textit{Table \ref{tab:prob_alltests_multi}} in which $s*$ denotes the set of all scores starting in $s$.

\begin{table}[H]
\begin{tabular} {|c|c|c|c|c|} \hline
Score / Type & $\bm{(1,H)}$ & $\bm{(1,L)}$ & $\bm{(2,H)}$ & $\bm{(2,L)}$ \\ \hline
 $\bm{A}$ 
 & $\alpha \phi p$  &$\bar{\alpha} \phi \bar{p}$ & $\alpha\bar{\phi}pf_H(A)$ & $\bar{\alpha}\bar{\phi}\bar{p}f_L(A)$     \\ \hline
 $\bm{B}$ 
 & $\bar{\alpha} \phi p$ & $\alpha \phi \bar{p}$ & $\bar{\alpha}\bar{\phi}pf_H(B)$ &$\alpha\bar{\phi}\bar{p}f_L(B)$  \\ \hline
 $\bm{AA*}$ 
 &0 &0 &$\alpha^2\bar{\phi}p\bar{f}_H(A)$ &$\bar{\alpha}^2\bar{\phi}\bar{p}\bar{f}_L(A)$   \\ \hline
$\bm{AB*}$ 
 & 0&0 &$\alpha\bar{\alpha}\bar{\phi}p\bar{f}_H(A)$ & $\alpha\bar{\alpha}\bar{\phi}\bar{p}\bar{f}_L(A)$  \\ \hline
 $\bm{BA*}$ 
 & 0&0 &$\alpha\bar{\alpha}\bar{\phi}p\bar{f}_H(B)$ & $\alpha\bar{\alpha}\bar{\phi}\bar{p}\bar{f}_L(B)$  \\ \hline
 $\bm{BB*}$ 
 &0 &0 & $\bar{\alpha}^2\bar{\phi}p\bar{f}_H(B)$& $\alpha^2\bar{\phi}\bar{p}\bar{f}_L(B)$  \\ \hline
 \textit{Total }
 & $\phi p$ & $\phi \bar{p}$ & $\bar{\phi}p$ & $\bar{\phi}\bar{p}$  \\ \hline
\end{tabular}
\caption{\label{tab:prob_alltests_multi}Distribution of Testing Outcomes with k Tests Available}
\end{table}

    From \textit{Table \ref{tab:prob_alltests_multi} }
    %\ar{Wait, huh? Table \ref{tab:prob_alltests} is for the $k=2$ case, what does it tell us here?}
    %\mn{Another way to read Table \ref{tab:prob_alltests} is to summarize the fractions for one score and more than one score (thus we can just see $AA$ as $AA*$, $BA$ as $BA*$, etc.). I will introduce a new table  for multiple tests that resembles Table \ref{tab:prob_alltests} to put things more clearly.}
    , we can compute:
\begin{align*}
    p_A &= \frac{\phi p \alpha + \bar{\phi} p \alpha f_H(A)}{\phi(p \alpha + \bar{p}\bar{\alpha}) + \bar{\phi}( p \alpha f_H(A) + \bar{p}\bar{\alpha} f_L(A))} \in (0, 1),\\
    p_B &= \frac{\phi p \bar{\alpha} + \bar{\phi} p \bar{\alpha} f_H(B)}{\phi(p \bar{\alpha} + \bar{p}\alpha) + \bar{\phi}( p \bar{\alpha} f_H(B) + \bar{p}\alpha f_L(B))} \in (0, 1), \\
    p_{AA*} &= \frac{p \alpha^2 \bar{f}_H(A)}{p \alpha^2 \bar{f}_H(A) + \bar{p}\bar{\alpha}^2\bar{f}_L(H)} \mathbbm{1} \{\bar{f}^2_H(A) + \bar{f}^2_L(A)>0\},\\
    p_{AB*} &= \frac{p \bar{f}_H(H)}{p \bar{f}_H(A) + \bar{p}\bar{f}_L(H)} \mathbbm{1}\{ \bar{f}^2_H(A) + \bar{f}^2_L(A)>0\}, \\
    p_{BA*} &= \frac{p \bar{f}_H(B)}{p \bar{f}_H(B) + \bar{p}\bar{f}_L(B)} \mathbbm{1}\{\bar{f}^2_H(B) + \bar{f}^2_L(B)>0\},\\
    p_{BB*} &= \frac{p \bar{\alpha}^2 \bar{f}_H(B)}{p \bar{\alpha}^2 \bar{f}_H(B) + \bar{p}\alpha^2\bar{f}_L(B)}\mathbbm{1}\{\bar{f}^2_H(B) +\bar{f}^2_L(B)>0\}
\end{align*}
where $\mathbbm{1}\{\cdot\}$ is the indicator function whose value is 1 if the statement in brackets is true and 0 otherwise. Note that for any $s \in \{A,B\}, p_{sB*} \leq p_{sA*}$ if $\bar{f}^2_H(s) + \bar{f}^2_L(s)>0$ (i.e., some students take the test more than once when the first score is $s$). Furthermore, to identify the relationship between the value of $p_{s*}$ and the matching equilibrium admissions decisions for a reported score sequence $s$, we show the following useful lemma.

\begin{lemma}\label{lem:conti}
If a score sequence $s \in \cup_{i = 1}^{k}\{A,B\}^i$ yields admission on the equilibrium path\footnote{In other words, $s$ is obtained with positive probability in equilibrium.},  then $p_{s*} \geq \frac{1}{2}$. Additionally, if $p_{s*} > \frac{1}{2}$, the score sequence $s$ yields admission for $s \in \{A,B\}$.
\end{lemma}

\begin{proof}
If all score sequences starting with $s$ yield admission, then by \textit{the law of total probability}, $p_{s*} \geq \frac{1}{2}$ as desired. 

Otherwise, the college rejects some score sequence starting in $s$ after some number of tests. We write $T$ to denote the length of the longest continuation of $s$ that \emph{cannot} result in rejection i.e., 
$$T = |s| + \max\{m \in \mathbbm{N} | u_{s\Tilde{s}} =1, \forall \Tilde{s} \in \cup_{i=1}^m \{A,B\}^i\}$$ 
where $|s|$ denotes the length of the score sequence $s$. If the set $\{m \in \mathbbm{N} | u_{s\Tilde{s}} =1, \forall \Tilde{s} \in \cup_{i=1}^m \{A,B\}^i\}$ is empty, then $T = |s|$.  After having taken $T$ tests, students with score sequences starting with $s$ would not take any further test, since if they stop, they are guaranteed admissions, but if they continue they risk rejection.  Therefore, in equilibrium we have $p_{s*} \geq \frac{1}{2}$. 

The proof above gives us a necessary condition for $s$ to yield admission on the equilibrium path, that is, $p_{s*} \geq \frac{1}{2}$\footnote{This condition alone may not suffice to tell us whether or not $s$ yields admission in equilibrium. For example, if $p_{BB*} \geq \frac{1}{2}$ (i.e. $p \geq \frac{\alpha^2}{\alpha^2 + \bar{\alpha}^2}$), there may exist an equilibrium in which only $A$ and $BBA$ yield admission. Hence, $p_{BB*} \geq \frac{1}{2}$ alone is not sufficient to derive that $BB$ is admitted in equilibrium.}. Now we investigate a special case of single scores, and characterize the sufficient condition for $s \in \{A,B\}$ to be admitted. We will show by contradiction that $s$ yields admission given $p_{s*} > \frac{1}{2}$ for some $s \in \{A,B\}$. If, otherwise, a length-one score sequence $s$ yields rejection, then we have the following two cases: (i) if the College rejects all scores starting in $s$, then by the \textit{law of total probability}, $p_{s*} \leq \frac{1}{2}$, a contradiction; (ii) if the College admits students submitting scores starting in $s$ for \emph{some} continuation, then all category 2 students take the test more than once if their first score is $s$, and thus $p_{s} = p_{s*} > \frac{1}{2}$, a contradiction again. Therefore, as long as $p_{s*} > \frac{1}{2}$, $s\in \{A,B\}$ yields admission. 
\end{proof}

With this lemma in hand,  we are now in position to discuss the equilibrium outcomes when $k$ tests are available. By \textit{Lemma \ref{lem:conti}}, a single score $A$ yields admission (i.e., $u_A = 1$) if $p>\bar{\alpha}$ and only if $p \geq \bar{\alpha}$ (i.e., $p_{A*} = \frac{p\alpha}{p\alpha + \bar{p} \bar{\alpha}} \geq \frac{1}{2}$). A single score $B$ yields admission (i.e., $u_B = 1$) if $p > \alpha$ and only if $p \geq \alpha$ (i.e., $p_{B*} = \frac{p\bar{\alpha}}{p\bar{\alpha} + \bar{p}\alpha} \geq \frac{1}{2}$). Therefore, for a non-trivial admission outcome to depend only on the first score, we have $p \in [\bar{\alpha}, \alpha]$. %To see why this condition is also sufficient for the existence of the first-score-based admission outcome, we have an equilibrium in which the College only admits the single score $A$ and all students take test once for any $p\in [\bar{\alpha}, \alpha]$.

Now we characterize the condition for a non-first-score equilibrium to exist. Suppose a single $B$ yields admission, then $p_{A*} > p_{B*} \geq \frac{1}{2}$, and thus a single $A$ also yields admission and all students are admitted. Therefore, for an equilibrium admission rule to depend on more than one score, we must have that a single $B$ yields rejection, i.e. $p \leq \alpha$. Furthermore, for some subsequent score after $B$ to matter in the admission outcome, that is, to make some score sequence starting in $B$ yield admission, we need $\frac{p \bar{\alpha} \alpha^{k-1}}{p \bar{\alpha} \alpha^{k-1} + \bar{p} \alpha \bar{\alpha}^{k-1}} \geq \frac{1}{2}$. To see why, we discuss the following two cases:
\begin{enumerate}
    \item If either $BA$ or $BB$ gives admission, then Category 2 students whose first score is $B$ always take a second test, and thus by \textit{Lemma \ref{lem:conti}}, $ \max\{p_{BA*}, p_{BB*}\} = p \geq \frac{1}{2}$. Then we have $\frac{p \bar{\alpha} \alpha^{k-1}}{p \bar{\alpha} \alpha^{k-1} + \bar{p} \alpha \bar{\alpha}^{k-1}} \geq p \geq \frac{1}{2}$ as desired.
    
    \item If, otherwise, both $BA$ and $BB$ give rejection, then denote by $M$ the maximum length of a continuation of tests such that \emph{every} continuation of length $\leq M$ results in rejection\footnote{The maximum exists since the set $\{m \in \mathbbm{N} | u_{Bs} =0, \forall s \in \cup_{i=1}^m \{A,B\}^i\}$ is nonempty (the College rejects both $BA$ and $BB$) and compact (the number of scores is an integer with the upper bound $k-1$).}, i.e.,
$$M = \max\{m \in \mathbbm{N} | u_{Bs} =0, \forall s \in \cup_{i=1}^m \{A,B\}^i\}.$$
Observe that we have $1 \leq M \leq k-2$ if some score beginning in $B$ yields admission. Students who take tests no more than $M+1$ times with a first score of $B$ are rejected. By definition, there exists $s \in \{A,B\}$ such that $B \Tilde{s} s$ gives admission for some $\Tilde{s} \in \cup_{i=1}^M \{A,B\}^i $. Thus $f_H(B\Tilde{s}) = f_L(B\Tilde{s}) = 0$, and thus $p_{B\Tilde{s}s*} \geq \frac{1}{2}$ by \textit{Lemma \ref{lem:conti}}. Let $G$ be the number of $A$ scores in $B\Tilde{s}s$ (Note that $0 \leq G \leq M+1$). We have $\frac{p \bar{\alpha} \alpha^{k-1}}{p \bar{\alpha} \alpha^{k-1} + \bar{p} \alpha \bar{\alpha}^{k-1}} \geq \frac{p \bar{\alpha} \alpha^{M+1}}{p \bar{\alpha} \alpha^{M+1} + \bar{p} \alpha \bar{\alpha}^{M+1}} \geq p_{BMs*} = \frac{p \bar{\alpha}^{M+2-G} \alpha^{G}}{p \bar{\alpha}^{M+2-G} \alpha^{G} + \bar{p} \alpha^{M+2-G} \bar{\alpha}^{G}} \geq \frac{1}{2}$.
\end{enumerate}
Therefore, $\frac{p \bar{\alpha} \alpha^{k-1}}{p \bar{\alpha} \alpha^{k-1} + \bar{p} \alpha \bar{\alpha}^{k-1}} \geq \frac{1}{2}$ (i.e., $p \geq p^*_k \equiv \frac{\bar{\alpha}^{k-2}}{\alpha^{k-2} + \bar{\alpha}^{k-2}}$) is a \textit{necessary} condition for some subsequent score after $B$ to affect the admission outcome.

To see why $p \in [p^*_k, \alpha]$ is also \textit{sufficient} for a non-first-score equilibrium to arise under ``Report All", define a partition $\{{p}^*_n\}_{n = 1}^{k}$ over the range $[p^*_k, \alpha]$ in which ${p}^*_n = \frac{\bar{\alpha}^{n-2}}{\alpha^{n-2} + \bar{\alpha}^{n-2}}$ for all $n = 1, \cdots, k$ \footnote{Note that ${p}^*_1 = \alpha$, ${p}^*_2 = \frac{1}{2}$, $\cdots$, $p^*_k \equiv \frac{\bar{\alpha}^{k-2}}{\alpha^{k-2} + \bar{\alpha}^{k-2}}$, and $\cup_{n=1}^{k} [p^*_n, p^*_{n+1}] = [p^*_k,\alpha]$.}. If $p \in [p^*_{n-1}, p^*_{n}]$, there is a non-first-score equilibrium in which the College accepts students with first score $A$, and rejects all students with first score $B$ unless their sequence of scores is  $B\underbrace{A \hdots A}_{n-1}$. In this equilibrium, Category 2 students take at least two tests if their first score is $B$, and both score sequences starting with $A$ and those realizing $B\underbrace{A \hdots A}_{n-1}$ yield admission. 

Similarly, we can derive the sufficient and necessary condition for some subsequent score after $A$ to matter in the admission outcome, that is, $p \in [\frac{\bar{\alpha}^{k}}{\alpha^{k} + \bar{\alpha}^{k}}, \bar{\alpha}]$. In this case, students with a single score are rejected.

Finally, we can study possible equilibrium outcomes if $p \in (\bar{\alpha}, \alpha)$. By Lemma \ref{lem:conti}, if $p > \bar{\alpha}$ (i.e., $p_{A*} > \frac{1}{2}$), a single $A$ yields admission; and if $p < \alpha$ (i.e., $p_{B*} < \frac{1}{2}$), a single $B$ yields rejection.  What remains to be proved is that any score sequence consisting only of $B$'s leads to rejection. Denote\footnote{$\cup_{i=1}^m \{B\}^i = \{B, BB, BBB, \cdots, \underbrace{B \hdots B}_{m}\}.$}
$$M_B = \max\{m \in \mathbbm{N} | u_{s} =0, \forall s \in \cup_{i=1}^m \{B\}^i\}.$$
Note that  we have $1 \leq M_B \leq k$ by definition. To show $M_B = k$, suppose by contradiction that $M_B \leq k-1$. By definition, students with score sequence $\underbrace{B\hdots B}_{M_B}$ are  rejected and students with score sequence $\underbrace{B\hdots B}_{M_B + 1}$ are admitted. Therefore Category 2 students who have thus far only obtained $B$ scores won't stop until they take $M_B + 1$ tests. Therefore we have 
$$p_{\underbrace{B\hdots B}_{M_B + 1}} = \frac{\bar{\alpha}^{M+1} p}{\bar{\alpha}^{M+1} p + \alpha^{M+1} \bar{p}} < \frac{\bar{\alpha}^{M} p}{\bar{\alpha}^{M} p + \alpha^{M} \bar{p}} = p_{\underbrace{B\hdots B}_{M_B}} \leq \frac{1}{2},$$
a contradiction.
%In sum, if $p > \alpha$ or $p < \min\{1-\alpha, p^*_k\}$, there exist trivial equilibria under ``Report All" in which all students are admitted or all are rejected regardless of their score. The sufficient and necessary condition for the existence of the first-score-based admission equilibrium is $p \in [1-\alpha,\alpha]$. In this equilibrium, all students whose first score is $A$ are accepted and all students whose first score is $B$ are rejected. This is the only possible nontrivial equilibrium if $p < p^*_k$. If, otherwise, $p \in [p^*_k, \alpha]$, then in addition to  the first-score based equilibrium, there is another equilibrium outcome in which the College admits student whose first score is $A$ or whose score sequence is $B \underbrace{A \hdots A}_{n-1}$, and reject all the others for some $2 \leq n \leq k$. \ar{So the first score equilibrium and the ``extreme case'' are the only two equilibria?}\mn{No. There can be other equilibria in which the admission outcome is not only based one the first score. But whenever such equilibria are possible, there must be an equilibrium that is characterized by the extreme case, which is to accept a score $B$ followed by $k-1$ $A$. More rigorous proof on this point is in Appendix A as indicated in footnote 13.}
\end{proof}

A direct result of \textit{Theorem \ref{th: all_eqm}} is that when $p \in(\bar{\alpha}, p^*_k)$, the admission outcome is unique and it depends only on the first score. Note that $ p^*_k > \bar{\alpha}$ only when $k = 2$. This gives us the following corollary, which is a generalization of \textit{Theorem \ref{th:all_2_eqm}}. 

\begin{corollary} \label{cor:report all_multi}
If $k = 2$ and $p \in (1-\alpha, p^*_2) = (1-\alpha, \frac{1}{2})$, the unique equilibrium under ``Report All" is the first-score equilibrium.
If $k \geq 3$, a first-score equilibrium always coexists with an equilibrium in which the admission outcome can be affected by more than the first score.
% , i.e., $u_A=u_{AB}=u_{AA}=1$ and $u_B=u_{BA}=u_{BB}=0$.
\end{corollary}

\subsection{Comparisons}
%It is straightforward to compare the false positive and false negative rates of the separating equilibria under `Report Max''  with that of ``Report All'' in the parameter range in which there is a unique nontrivial equilibrium in both cases. 

%We summarize the results in the tables below and we can see that in a corresponding parameter range in which $p \in [\hat{p}_{k}, p^*_k) $\footnote{``Report Max" has the unique nontrivial equilibrium when $p \in [\hat{p}_{k}, \hat{p}'_{k}]$ by \textit{Theorem \ref{th: max_multi_eqm}}. ``Report All" has the unique nontrivial equilibrium when $p \in (1-\alpha, \hat{p}^*_{k})$ by \textit{Corollary \ref{cor:report all_multi}}. Since $\hat{p}'_{k} > \alpha > \frac{1}{2} \geq  p^*_k, \hat{p}_{k} > 1-\alpha$, the common range in which unique nontrivial equilibrium exists under both schemes is $p \in [\hat{p}_{k}, p^*_k)$. Note that $\hat{p}_k$ strictly increases in $k$ and $p^*_k$ strictly decreases in $k$. Thus, this parameter range of interest shrinks as the number of tests available increases. Furthermore, since $p^*_2 - \hat{p}_2 =\frac{1}{2} - \hat{p} > 0$ and $\lim_{k \rightarrow \infty}(p^*_k - \hat{p}_k) = 0 - \frac{\phi \bar{\alpha} + \bar{\phi}}{1 + \bar{\phi}} < 0$, there exists a maximal number of tests, $K(\alpha,\phi) > 2$ for each given $\alpha,\phi$. If $k \leq K(\alpha,\phi))$, this range is nontrivial.}

Once again, it is straightforward to compare the false positive and false negative rates of the separating equilibria under `Report Max''  with the first-score equilibirum under ``Report All'', whenever it exists. The relevant parameter range is $p \in [\hat{p}_{k}, \hat{p}'_{k}]\cap [\bar{\alpha}, \alpha] = [\hat{p}_{k}, \alpha]$ which is nontrivial for all $k\geq 2$ since $\hat{p}_{k} < \frac{1}{2} < \alpha$. We summarize the results in the tables below and we can see that in a corresponding parameter range in which $p \in [\hat{p}_{k}, \alpha]$, ``Report All" achieves parity across categories in false positive and false negative rates, whereas for ``Report Max'', there is necessarily a discrepency between false positive rates and false negative rates---in favor of the advantaged (Category 2) students---which increases as the number of tests $k$ available to the advantaged population increases.

\begin{table}[H]
    \centering
    \begin{tabular} {|c|c|c|c|} \hline
 & $\bm{(1,H)}$  & $\bm{(2,H)}$  \\ \hline
 \textbf{Max} & $1-\alpha$  &  $(1-\alpha)^k$     \\ \hline
\textbf{All } & $1-\alpha$ &   $1-\alpha$   \\ \hline
\end{tabular}
\quad
   \begin{tabular} {|c|c|c|c|} \hline
 & $\bm{(1,L)}$  & $\bm{(2,L)}$  \\ \hline
 \textbf{Max} & $1-\alpha$  &  $1-\alpha^k$     \\ \hline
\textbf{All} & $1-\alpha$ &   $1-\alpha$   \\ \hline
\end{tabular}
    \caption{\label{tab:FNFP_multi} \footnotesize False Negative (left) and False Positive (right) Rates with k Tests Available [Using First-Score Equilibrium]} 
\end{table}
\normalsize

What about equilibria other than the first-score equilibrium? Here, students from Category 2 again have an advantage, but we can quantify the \emph{degree} of that advantage by measuring the disparity between false positive rates and false negative rates between Category 1 and Category 2 students. What we can show is that this disparity is always lower under any equilibrium of ``Report All'' compared to ``Report Max'' -- and hence, while \emph{some} inequity may remain under ``Report All'', it is  reduced (always weakly, and strictly if either $k > 3$ or $p < \frac{1}{2}$) compared to ``Report Max''. 

We proceed by comparing the unique nontrivial ``Report Max'' equilibrium to an equilibrium under ``Report All''. First note that under both ``Report All" and ``Report Max", a single $A$ yields admission while a single $B$ yields rejection when $p \in [\hat{p}_k, \alpha]$. Therefore, the false positive and false negative rates remain the same for Category 1 students across both admissions policies. Next, we note that under any equilibrium under ``Report All'', a score sequence consisting entirely of $B$'s results in rejection --- and hence the admissions probability for Category 2 students (regardless of their type) can only be \emph{smaller} under Report All than under Report Max --- since any other sequence of scores would lead to admission under Report Max, but \emph{possibly} not under Report All. Finally, we observe that if either $k > 3$ or $p < 1/2$, there is \emph{some} sequence of scores that leads to rejection under Report All, but not under Report Max, showing that the probability of admission for Category 2 students is \emph{strictly} lower under Report All for both types, which implies the corresponding reduction in false positive and false negative disparities. To see, consider that if this does not hold, the College must accept students with the length-$k$ score sequence $\hat{s} = B\hdots B A$ while rejecting all score sequences $s \in \cup_{i = 1}^{k}\{B\}^{i}$. Under this admissions policy, Category 2 students who have only received $B$'s thus far would continue to take the test (up to $k$ times) and thus to rationalize the admissions rule, we must have: $\frac{p \bar{\alpha}^{k-1}\alpha}{p \bar{\alpha}^{k-1}\alpha + \bar{p} \alpha^{k-1}\bar{\alpha}} \geq \frac{1}{2}$, i.e., $p \geq \frac{\alpha^{k-1} \bar{\alpha}}{\alpha^{k-1} \bar{\alpha} + \bar{\alpha}^{k-1} \alpha} \geq \frac{1}{2}$. Note that a nontrivial equilibrium under ``Report All" exists only if $p \leq \alpha$, and $\frac{\alpha^{k-1} \bar{\alpha}}{\alpha^{k-1} \bar{\alpha} + \bar{\alpha}^{k-1} \alpha} \leq \alpha$ only when $k\in \{2,3\}$. Hence, if $p < \frac{1}{2}$ or $k>3$, there exists at least one score sequence ($\hat{s}$) which is rejected  under ``Report All'' but is admitted under ``Report Max''.
%A rigorous proof of this point can be seen in \textit{Appendix A}. 

As in the $k=2$ case, we similarly find that in the first-score  equilibrium under ``Report All'', the positive predictive value is strictly higher, compared to the separating equilibrium under ``Report Max'' --- i.e. the admitted class has a higher proportion of High types. 
 \begin{theorem}\label{th:PPV_multi}
    For any $\alpha \in (\frac{1}{2}, 1)$, the positive predictive value in the first-score equilibrium under the ``Report All'' policy exceeds that under the ``Report Max'' separating equilibrium.
    \end{theorem}
    \begin{proof} 
    %For the sake of comparison, we need $p \in [\hat{p}_k,\alpha]$ such that equilibrium is nontrivial under both admission policies\footnote{For $p \in [0,\hat{p}_k) \cup (\alpha, 1]$, the statement in \textit{Theorem \ref{th:PPV_multi}} holds vacuously true.}. 
    The positive predictive value of ``Report Max'' is 
$
 \frac{\alpha \phi p + (1-\bar{\alpha}^k)\bar{\phi} p}{\alpha \phi p  + \bar{\alpha} \phi \bar{p} +(1- \bar{\alpha}^k)\bar{\phi}p+ (1-\alpha^k)\bar{\phi}\bar{p} }.$ The positive predictive value  of the first-score equilibrium under ``Report All" is $\frac{\alpha p}{\alpha p + \bar{\alpha}\bar{p}}.$ It is straightforward to verify that the first term is strictly smaller than the second if $\alpha \in (0.5, 1)$.
%In equilibrium, where the College admits students whose first score is $A$ or whose score is $L \underbrace{H \cdots H}_{k-1}$, and rejects others, the positive predictive value is 
%$$\frac{p \phi \alpha + p \bar{\phi} ( \alpha+ \bar{\alpha}\alpha^{k-1})}{p \phi \alpha + p \bar{\phi} ( \alpha+ \bar{\alpha}\alpha^{k-1}) + \bar{p} \phi \bar{\alpha} + \bar{p} \bar{\phi} ( \bar{\alpha}+ \alpha \bar{\alpha}^{k-1})}$$
%which is even higher than the second term, and thus is higher than the first term if $\alpha \in (0.5, 1)$.
\end{proof}

Similarly, as before, the negative predictive value is higher under ``Report Max'':

 \begin{theorem}
    For any $\alpha \in (\frac{1}{2}, 1)$, the negative predictive value in the first-score equilibrium under the ``Report All'' policy is  strictly smaller than that of the ``Report Max'' separating equilibrium.
    \end{theorem}
    \begin{proof}
The negative predictive value of ``Report Max'' is
$\frac{\alpha \phi \bar{p} + \alpha^k\bar{\phi}\bar{p}}{\bar{\alpha}\phi p + \alpha \phi \bar{p} + \bar{\alpha}^k\bar{\phi}p + \alpha^k\bar{\phi}\bar{p}}.$ The negative predictive value of the first-score equilibrium under``Report All'' is
$\frac{\bar{p}\alpha}{p\bar{\alpha} + \bar{p}\alpha}.$ It is straightforward to verify that the first term is larger than the second for $\alpha \in (\frac{1}{2},1)$.
\end{proof}

Finally, we recover that even in the general case, the College's utility is strictly higher under the ``Report All'' policy  compared to ``Report Max'' --- and so once again, the ``Report All'' policy is preferable not just from the perspective of equity, but also from the selfish perspective of the College. First, we prove a lemma that establishes this specifically for the first-score equilibrium of ``Report All'' --- but we use this Lemma to prove \textit{Theorem \ref{thm:gen-payoff}} below, which establishes the result for \emph{all} equilibria under ``Report All''. 

\begin{lemma}\label{th: payoff_multi}
   Let $p^{**}_k = \frac{\alpha - \alpha^k}{1 - \alpha^k - \bar{\alpha}^k} \in [\frac{1}{2}, \alpha)$ For any $\alpha \in (\frac{1}{2}, 1)$, $p < p^{**}_k$, the College's expected payoff in the first score equilibrium under the ``Report All'' policy exceeds what it is under the ``Report Max'' separating equilibrium. The expected payoff gap increases with the number of tests $k$ available to Category 2 students.
\end{lemma}

\begin{proof}
The expected payoff per student under the first score  equilibrium for ``Report All'' is $\alpha p - \bar{\alpha}\bar{p}$. The expected payoff per student under ``Report Max'' is
$\phi(\alpha p - \bar{\alpha}\bar{p}) + \bar{\phi}[(1-\bar{\alpha}^k)p - (1-\alpha^k)\bar{p}].$ Hence, the difference between the College's expected payoff under these two schemes is:
\begin{align*}
&\alpha p - \bar{\alpha}\bar{p} - \{\phi(\alpha p - \bar{\alpha}\bar{p}) + \bar{\phi}[(1-\bar{\alpha}^k)p - (1-\alpha^k)\bar{p}]\}   \\
=\quad &\bar{\phi}[\alpha p - \bar{\alpha}\bar{p} -  (1-\bar{\alpha}^k)p + (1-\alpha^k)\bar{p}]\\
=\quad &\bar{\phi} \alpha \bar{\alpha} \sum_{i =0}^{k-2}(\alpha^i \bar{p} - \bar{\alpha}^i p) =  \bar{\phi} [(\alpha - \alpha^k) \bar{p} - (\bar{\alpha} - \bar{\alpha}^k) p] > 0 
\end{align*}
where the last inequality holds when $p < \frac{\alpha - \alpha^k}{1 - \alpha^k - \bar{\alpha}^k} = p^{**}_k$ and $\alpha > \frac{1}{2}$.
\end{proof}

We can strengthen \textit{Lemma \ref{th: payoff_multi}} by using it to prove that for \emph{every} nontrivial equilibrium under ``Report All'', the College has higher expected payoff  compared to ``Report Max''. Hence, the College has an incentive to prefer this policy even without the ability to perform equilibrium selection.

\begin{theorem}
\label{thm:gen-payoff}
    For any $\alpha \in (\frac{1}{2}, 1)$, $p < p^{**}_k$\footnote{Note that $p^{**}_k= \frac{\alpha - \alpha^k}{1 - \alpha^k - \bar{\alpha}^k} \in [\frac{1}{2}, \alpha)$ strictly increases in $k$ and $\lim_{k \rightarrow \infty} p^{**}_k = \alpha$, so the range of interest expands as the number of tests available increase and in the limit, the results in \textit{Theorem \ref{thm:gen-payoff}} hold in every nontrivial equilibrium comparison for $p \in (0,1)$. Moreover, we may want to capture the idea that talent is scarce by further assuming that $p < \frac{1}{2} \leq p^{**}_k$. In this case, findings in \textit{Theorem \ref{thm:gen-payoff}} are valid.}, the College obtains strictly higher utility under the ``Report All" policy compared to the``Report Max" policy in \emph{every} nontrivial equilibrium. 
    
    Additionally, amongst the ``Report All" equilibria, the College has a higher expected payoff under equilibria in which the admissions outcome depends on more than the first reported score.
\end{theorem}
\begin{proof}
 The statement holds vacuously true for $p < \hat{p}_k \in (0, p_k^{**})$ since there only exists trivial equilibrium of rejecting all under ``Report Max". Recall that when $p \in [\hat{p}_{k}, p_k^{**}) \subset (1-\alpha, \alpha)$, a student whose first (or only) score is $A$ gets admission under both schemes. Therefore, the only difference between equilibria under ``Report All" and ``Report Max" consists in how they treat students whose first (or only) score is $B$. By Lemma \ref{th: payoff_multi}, we know already that for the first score equilibrium under ``Report All", the College strictly prefers ``Report All" to ``Report Max". 

Therefore, it remains to study the case in which the equilibrium admissions outcome depends on more than the first score. Fix any such an equilibrium and define $S$ to be the set of score sequences of length $\geq 2$ that lead to admissions\footnote{The set $S$ is nonempty. Otherwise we are back to the case in which only the first score matter for the admissions outcome.}:
$$ S \equiv \{s \in \cup_{i = 2}^{k} \{A,B\}^i | u_s = 1\}.$$
For any $s \in S$, denote the fraction of students who obtain score $s$ in the equilibrium by $q_s \in [0,\bar{\phi}]$. Then the expected payoff for the college under this equilibrium of the ``Report All'' policy is:
$$\underbrace{\alpha p - \bar{\alpha} \bar{p}}_{\text{$A$ yields admission}} + \underbrace{\sum_{s \in S} q_s [p_s + (-1) (1-p_s)]}_{\text{Longer score sequences yield admission}}.$$
For the College to admit $s \in S$, we need $p_s \geq \frac{1}{2}$, and thus the last term $\sum_{s \in S} q_s [p_s + (-1) (1-p_s)] = \sum_{s \in S} q_s (2 p_s -1) \geq 0$. Note that the first two terms are the expected payoff for the College in the first-score equilibrium. Therefore, the College prefers a non-first-score equilibrium than the first-score equilibrium under ``Report All", which in turn is strictly better than the separating equilibrium under ``Report Max" by \textit{Lemma \ref{th: payoff_multi}}.

\end{proof}

\section{Discussion}
Allowing students to retake standardized tests and report only the best scores obtained---a currently common practice known as ``super-scoring''---clearly gives an advantage to well-resourced students who have the ability to take the test multiple times. A natural fix would seem to be to require that all students take the exam only once, thereby enforcing equity---but for various reasons, including that tests are administered by third party entities with their own interests, and that different colleges have different admissions policies, this seems unworkable.  \textit{A priori}, the effects of a traditional  alternative---requiring students to report all of their scores --- are  less transparent. This is because it seems to still gives well-resourced students an advantage, as a population: it provides the option for the more talented students to report a more \emph{accurate} signal (by taking the exam several times), while allowing the less talented students to pool with the lower-resourced students by taking the exam only once, thereby providing a less accurate signal and an increased chance of admissions. 

Nevertheless, we show that \emph{in equilibrium}, the traditional policy of requiring that all scores be reported has the same effect as enforcing that students take the exam only once.  Moreover, this policy is preferable to super-scoring, both from the perspective of equity---in the ``Report All'' equilibrium, the chance that a student is admitted is independent of their population, conditional on their type---but also from the perspective of the college. This represents an unusual but important situation in which goals of accuracy and equity are in alignment. 

\paragraph*{Acknowledgements}
We thank Chris Jung, Changhwa Lee, and Mallesh Pai for helpful conversations at an early stage of this work. We gratefully acknowledge support from NSF grants CCF-1763307 and CCF-1763349 and the Simons Collaboration on the Theory of Algorithmic Fairness.

\bibliographystyle{alpha}
\bibliography{refs}

\appendix

\end{document}